\documentclass[a4paper,UKenglish]{article}

\usepackage{prooftree}
\usepackage{pdflscape}
\usepackage{amsmath}
\usepackage{amssymb}
\usepackage{url}
\usepackage{color}

\newtheorem{Theorem}{Theorem}          
\newtheorem{Lemma}[Theorem]{Lemma}

\newtheorem{Definition}{Definition}    
\def\Proof{\par \noindent{\bf Proof: }}
\def\Done{\hfill\rule{0.5em}{0.5em}}
\newenvironment{proof}{\Proof}{\Done}

\def\NN{\mathbb{N}}

\def\Lolly{\Rightarrow}
\def\iAnd{\otimes} 

\def\ASM{[\Func{ASM}]}
\def\CUT{[\Func{CUT}]}

\newcommand{\systemB}{{\cal B}}
\newcommand{\LambdaP}{\Lambda^{\otimes}}

\newcommand{\betaConv}{\textbf{$\beta$-conv}}
\newcommand{\letConv}{\textbf{l-conv}}
\newcommand{\breakConv}{\textbf{b-conv}}
\newcommand{\apLetConv}{\textbf{ap-l-conv}}
\newcommand{\apBreakConv}{\textbf{ap-b-conv}}
\newcommand{\letLetConv}{\textbf{l-l-conv}}
\newcommand{\letBreakConv}{\textbf{l-b-conv}}
\newcommand{\breakLetConv}{\textbf{b-l-conv}}

\def\LLeft{[{\Lolly}\Func{L}]}
\def\LRight{[{\Lolly}\Func{R}]}
\def\CLeft{[{\iAnd}\Func{L}]}
\def\CRight{[{\iAnd}\Func{R}]}

\def\BRK{[\mbox{\Func{BRK}}]}
\def\CE{[{\iAnd}\Func{E}]}
\def\CI{[{\iAnd}\Func{I}]}

\def\LE{[{\Lolly}\Func{E}]}
\def\LI{[{\Lolly}\Func{I}]}
\def\LLm{\Logic{{\L}L}{m}}

\def\FV{\textbf{FV}}

\def\Logic#1#2{\mbox{{\bf #1}}_{\mbox{\bf #2}}}

\newcommand{\ALm}{\Logic{AL}{m}}
\newcommand{\LLi}{\Logic{{\L}L}{i}}
\newcommand{\LLc}{\Logic{{\L}L}{c}}

\def\Pair#1#2{{#1} \mathop{\mbox{{\scriptsize $\otimes$}}} {#2}}

\def\Let#1#2#3{\mbox{\textbf{let}}\,#2 \mathrel{\mbox{\textbf{be}}} #1 \mathrel{\mbox{\textbf{in}}} #3}
\newcommand{\Brk}{\mbox{\textbf{break}}}
\newcommand{\Lets}{\mbox{\textbf{let}}}
\newcommand{\Break}[4]{\Brk_{#1}\:#2\;\mbox{\textbf{as}}\;#3\;\mbox{\textbf{in}}\,#4}

\newcommand{\reduc}{\rightsquigarrow}
\newcommand{\reducstar}{\reduc^{*}}
\newcommand{\reducby}[1]{\stackrel{#1}{\reduc}}

\def\Func#1{{\sf{#1}}}

\def\Diff{\mathop{\backslash}}

\def\Hide#1{\relax}

\title{A Curry-Howard Correspondence for the Minimal Fragment of {\L}ukasiewicz Logic}
\author{Rob Arthan\thanks{
\texttt{rda@lemma-one.com}
} \& Paulo Oliva\thanks{
\texttt{p.oliva@qmul.ac.uk}
}}

\begin{document}

\maketitle

\begin{abstract}
In this paper we introduce a term calculus $\systemB$ which adds to the affine $\lambda$-calculus with pairing a new construct allowing for a restricted form of contraction. We obtain a Curry-Howard correspondence between $\systemB$ and the sub-structural logical system which we call ``minimal {\L}ukasiewicz logic'', also known in the literature as the logic of hoops (a generalisation of MV-algebras). This logic lies strictly in between affine minimal logic and standard minimal logic. We prove that $\systemB$ is strongly normalising and has the Church-Rosser property. We also give examples of terms in $\systemB$ corresponding to some important derivations from our work and the literature. Finally, we discuss the relation between normalisation in $\systemB$ and cut-elimination for a Gentzen-style formulation of minimal {\L}ukasiewicz logic.
 \end{abstract}

\section{Introduction}

We are interested in the proof theory of {\L}ukasiewicz logic and subsystems thereof.
In this context, designing proof systems with nice dynamical properties -- cut-elimination or normalisation -- has proved to be a hard problem. Systems with the cut-elimination property have been successfully obtained via an extension of Gentzen's sequent calculus by means of \emph{hypersequents} \cite{Metcalfe(2005A)}. But this approach depends on the pre-linearity axiom
\[ (A \Lolly B) \vee (B \Lolly A) \]
a principle that is not intuitionistically acceptable. As far as we are aware, in the quite extensive literature on fragments of {\L}ukasiewicz logic that are compatible with intuitionistic or minimal logic, such as the logic of GBL algebras \cite{jipsen-montagna06} and the logic of hoops \cite{blok-ferreirim00,BO}, no normalising or cut-free proof systems are to be found. 

Benton et al. \cite{Benton(1993)} gave a term calculus for intuitionistic linear logic.
In this paper, we follow their approach and extend the simply-typed affine $\lambda$-calculus with a construct that captures propositional \emph{minimal {\L}ukasiewicz logic}. We prove this extended system $\systemB$ preserves types, is strongly normalising and has the Church-Rosser property. We give examples of terms in $\systemB$ corresponding to some important derivations from our work and from the literature on hoops and GBL algebras. 

\subsection{Fragments of {\L}ukasiewciz logic}

The standard Hilbert-style axiomatisation of (classical) {\L}ukasiewciz logic $\LLc$ may be found in \cite{Hajek98}. It has {\it modus ponens} as its only inference rule and its axioms comprise the axioms of \emph{basic logic}:
\[
\begin{array}{ll}
	\textup{(B1)} & (A \Lolly B) \Lolly (B \Lolly C) \Lolly A \Lolly C \\[1mm]
	\textup{(B2)} & A \iAnd B \Lolly A \\[1mm]
	\textup{(B3)} & A \iAnd B \Lolly B \iAnd A \\[1mm]
	\textup{(B4)} & A \iAnd (A \Lolly B) \Lolly B \iAnd (B \Lolly A) \\[1mm]
	\textup{(B5a)} & (A \iAnd B \Lolly C) \Lolly A \Lolly B \Lolly C \\[1mm]
	\textup{(B5b)} & (A \Lolly B \Lolly C) \Lolly A \iAnd B \Lolly C \\[1mm]
	\textup{(B6)} & ((A \Lolly B) \Lolly C) \Lolly ((B \Lolly A) \Lolly C) \Lolly C \\[1mm]
	\textup{(B7)} & \bot \Lolly A
\end{array}
\]
together with the axiom of \emph{double negation elimination}.
\[
\begin{array}{ll}
	\textup{(DNE)} & \neg \neg A \Lolly A \\[1mm]
\end{array}
\]
where $\neg A$ is defined as $A \Lolly \bot$. If from these we drop the axioms that are not valid in minimal logic \cite{Troelstra(96)}, i.e. (B6), (B7) and (DNE), we are left with the fragment (B1)--(B5), which we will call \emph{minimal {\L}ukasiewicz logic} $\LLm$. If we extend $\LLm$ with the ex-falso-quodlibet axiom (B7), we obtain what we have called \emph{intuitionistic {\L}ukasiewicz logic} $\LLi$. The logics $\LLm$ and $\LLi$ can be faithfully characterised algebraically using the classes of algebraic structures originally due to B\"{u}chi and Owen and known as hoops and bounded hoops, respectively (see \cite{arthan-oliva14a,arthan-oliva14b,blok-ferreirim00,BO}). Hoops are reducts of commutative GBL algebras, whose equational theory has been shown to be PSPACE-complete \cite{Bova:2009}.
As the equational theory of GBL algebras is a conservative extension of that of hoops, it follows that the decision problems for $\LLm$ and $\LLi$ are also PSPACE-complete\footnote{We are indebted to the late Franco Montagna who pointed this out to us back in 2014.}.
If we omit (B4) from $\LLm$, we obtain the minimal $(\Lolly, \iAnd)$-fragment of affine logic $\ALm$, which we will just refer to as affine logic in this paper.
For more details about the proof theory of these systems, including double negation translations from $\LLc$ into $\LLi$ and $\LLm$, see \cite{arthan-oliva15a}.

Our goal in this paper is to devise a typed term calculus whose inhabited types comprise the formulas provable in $\LLm$ with $\iAnd$ and $\Lolly$ viewed as the product and function type constructors.

\section{The $\systemB$ Calculus}

Let us start by introducing the term language of the $\systemB$ calculus.
Since we want terms to be unambiguous representations of proofs, we give a Church-style calculus of typed terms, rather
than a Curry-style type-assignment system.

\subsection{Term language}

Types are formed from type variables $P_1, P_2, \ldots$ using the binary operators $\Lolly$ and $\iAnd$. We use $P$ to range over type variables and $A, B, C, D$ will range over arbitrary types. The terms of the $\systemB$ calculus are obtained inductively starting from typed variables ($x^A, y^B, \ldots$) via the following constructs:
\begin{itemize}
	\item $\lambda$-abstraction and term application
	\begin{itemize}
		\item $\lambda x^A.t$ is a term when $t$ is a term
		\item $s\,t$ is a term when $s$ and $t$ are terms
	\end{itemize}
	\item constructor and destructor for pairs
	\begin{itemize}
		\item $\Pair{s}{t}$ is a term when $s$ and $t$ are terms
		\item $\Let{\Pair{x^A}{y^B}}{t}{r}$ is a term when $r$ and $t$ are terms, $x^A$ and $y^B$ are distinct variables
	\end{itemize}
	\item the $\Brk$ constructor:
	\begin{itemize}
		\item $\Break{}{t}{x^A,y^B}{r}$ is a term when $r$ and $t$ are terms, $x^A$ and $y^B$ are distinct variables.
	\end{itemize}
\end{itemize}

Our typing rules will imply that the variables $x^A$ and $y^B$ in $\Break{}{t}{x^A,y^B}{r}$ denote a higher-order function and a function respectively. From now on, we will therefore generally use letters like $\varphi$ and $f$ for these variables instead of $x$ and $y$. This is illustrated in the following definition of the set $\FV(t)$ of free variables of a term $t$:
\begin{align*}
\FV(x^A) &= \{x^A\} \\
\FV(\lambda x^A . t) &= \FV(t) \Diff \{x^A\} \\
\FV(s\,t) &= \FV(s) \cup \FV(t) \\
\FV(\Pair{s}{t}) &= \FV(s) \cup \FV(t) \\
\FV(\Let{\Pair{x^A}{y^B}}{s}{t}) &= (\FV(t) \Diff \{x^A, y^B\}) \cup \FV(s) \\
\FV(\Break{}{s}{\varphi^A, f^B}{t}) &= (\FV(t) \Diff \{\varphi^A, f^B\}) \cup \FV(s)
\end{align*}


\subsection{Type system}
\label{sec:type-system}

\begin{figure*}
\[
\small{
\begin{array}{ccc}
%
%
\begin{prooftree}
\justifies
\Gamma, x : A\vdash x^A : A
\using{\ASM}
\end{prooftree}
& \quad &
\begin{prooftree}
\Gamma \vdash t : A
\qquad
\Delta, \varphi : K_B A, f : S_B A \vdash u : C
\justifies
\Gamma, \Delta \vdash \Break{}{t}{\varphi^{K_B A}, f^{S_B A}}{u} : C
\using{\BRK}
\end{prooftree}
\\[7mm]
%
%
\begin{prooftree}
\Gamma, x : A \vdash t : B
\justifies
\Gamma \vdash \lambda x^A . t : A \Lolly B
\using{\LI}
\end{prooftree}
& &
%
%
\begin{prooftree}
\Gamma \vdash t : A \Lolly B
\quad
\quad
\Delta \vdash u : A
\quad
\justifies
\Gamma, \Delta \vdash t\,u : B
\using{\LE}
\end{prooftree} \\[7mm]
%
%
\begin{prooftree}
\Gamma \vdash s : A
\quad
\Delta \vdash t : B
\justifies
\Gamma, \Delta \vdash \Pair{s}{t} : A \iAnd B
\using{\CI}
\end{prooftree}
& &
%
%
\begin{prooftree}
\Gamma \vdash t : A \iAnd B
\quad
\Delta, x : A, y : B \vdash u : C
\justifies
\Gamma, \Delta \vdash \Let{\Pair{x^A}{y^B}}{t}{u} : C
\using{\CE}
\end{prooftree} 
\end{array}
}
\]
\caption{$\systemB$ typing rules}
\label{fig-type-system}
\end{figure*}

The typing rules for $\systemB$ are given in Figure \ref{fig-type-system}.
In the sequents $\Gamma \vdash t \colon A$ used in the rules,
the context $\Gamma$ is a finite function mapping (untyped) variable names to types, $t$ is a $\systemB$ term and $A$ is a type.
In the rule \BRK, we use the following abbreviations\footnote{It is noteworthy that for each formula $B$ , the mappings $A \mapsto K_B A$ and $A \mapsto S_B A$ can both be equipped with a monad structure in the simply typed $\lambda$-calculus, with $A \mapsto K_B A$ being the well-known continuation monad.}:
\begin{align*}
 K_B A &\equiv (A \Lolly B) \Lolly B \\
   S_B A & \equiv \; A \Lolly B.
\end{align*}

The rules with two premises are subject to the side condition that the two contexts $\Gamma$ and $\Delta$
must be compatible, i.e. $\Gamma, \Delta$ must also be a finite function mapping (untyped) variable names to types. 
This implies that typable terms are {\em affine}
in the sense that in any subterm each free variable appears exactly once.
However, as we will see, with the new rule $\BRK$ we populate many
types that are uninhabited in the affine simply-typed $\lambda$-calculus with pairing.

We say \emph{a term $t$ is typable with type $A$} and write $t : A$ if the
rules of $\systemB$ allow us to infer a sequent of the form $\Gamma \vdash t :
A$, where $\Gamma$ consists of the mappings $x \colon B$ for each $x^B \in \FV(t)$. If $t : A$, then $A$ is uniquely determined by $t$. If one presents $\systemB$
as a Curry-style type assignment system, then the Milner-Hindley principal type algorithm \cite{Hindley97}
extends easily to $\systemB$ allowing us to find a most general type assignment for a given term. 
From now on we will adopt the usual conversion of omitting type superscripts that can be inferred from the context. 

\subsection{Correspondence between $\LLm$ and $\systemB$}

In this section we show that types inhabited by closed terms in $\systemB$ are precisely the provable formulas of $\LLm$. To see that any formula provable in the Hilbert-style system $\LLm$ is (when viewed as a type) inhabited by a closed term of $\systemB$, the main work is showing that the axioms of $\LLm$ are inhabited:

\begin{Theorem} If a formula $A$ is provable in $\LLm$ then there exists a closed $\systemB$ term $t$ which is typable with type $A$.
\end{Theorem}

\Proof The special case of $\LE$ where the contexts $\Gamma$ and $\Delta$
are empty provides us with {\it modus ponens}. Hence it is enough to show that we have closed terms whose types are precisely the axioms (B1)--(B5). We indeed have: \\[1mm]
(B1) $\lambda f \lambda g \lambda x . g(f(x))
	: (A \Lolly B) \Lolly (B \Lolly C) \Lolly A \Lolly C$ \\[1mm]
(B2) $\lambda v . \Let{\Pair{x}{y}}{v}{x}
	: A \iAnd B \Lolly A$ \\[1mm]
(B3) $\lambda v . \Let{\Pair{x}{y}}{v}{\Pair{y}{x}}
	: A \iAnd B \Lolly B \iAnd A$ \\[1mm]
(B4) $\lambda v . \Let{\Pair{x}{f}}{v}{(\Break{}{x}{\varphi, g}{\Pair{\varphi(f)}{g}})}
	: A \iAnd (A \Lolly B) \Lolly B \iAnd (B \Lolly A)$ \\[1mm]
(B5a) $\lambda f \lambda x \lambda y . f(\Pair{x}{y})
	: (A \iAnd B \Lolly C) \Lolly A \Lolly B \Lolly C$ \\[1mm]
(B5b) $\lambda g \lambda a . \Let{\Pair{x}{y}}{a}{g(x)(y)}
	: (A \Lolly B \Lolly C) \Lolly A \iAnd B \Lolly C$ \Done\\
%

For the converse, to see that any type inhabited by a closed term in $\systemB$ is a formula that is provable in the Hilbert-style system $\LLm$
requires a little ingenuity: we show that the logical principle embodied in the rule $\BRK$ is derivable from the axiom (B4).

\begin{Theorem} If $\Gamma \vdash t : A$ in $\systemB$, where $\Gamma \equiv x_1 \colon B_1, \ldots, x_n \colon B_n$, then $B_1 \Lolly \ldots \Lolly B_n \Lolly A$ is provable in $\LLm$.
\end{Theorem}

\Proof
If we erase the terms from the typing rules $\ASM$, $\LI$, $\LE$, $\CI$ and $\CE$, the resulting logical inference rules give a standard
natural deduction presentation of affine logic,
which is well known to be equivalent to the Hilbert-style system $\ALm$. So
we have only to show that the following rule, obtained by erasing the terms
from $\BRK$, is derivable in $\LLm = \ALm$ + (B4):
$$\begin{prooftree}
\Gamma \vdash A
\qquad
\Delta, K_B A, S_B A \vdash C
\justifies
\Gamma, \Delta \vdash C
\end{prooftree}
$$
That in turn, follows once we can show that $A \Lolly K_B A \iAnd S_B A$ is derivable in $\LLm$.
We start by noting that $\ALm$ proves
$\vdash A \Lolly (A \Lolly B) \Lolly B$
i.e., $\vdash A \Lolly K_B A$. Hence $\ALm$ proves
$\vdash A \Lolly (A \iAnd (A \Lolly K_B A))$.
But then, using axiom (B4) to transform $A \iAnd (A \Lolly K_B A)$, we have that $\LLm$ proves
$\vdash A \Lolly (K_B A\iAnd (K_B A \Lolly A))$.
As $\ALm$ also proves
$\vdash B \Lolly (A \Lolly B) \Lolly B$, i.e.,
$\vdash B \Lolly K_B A$, $\LLm$ proves
$\vdash A \Lolly K_B A \iAnd S_B A$.
I.e., in $\LLm$, $A$ is logically stronger than the conjunction
of $K_B A$ and $S_B A$, which justifies the inference given by
$\BRK$ \Done\\

In the above proof we have weakened $A \Lolly (K_B A\iAnd (K_B A \Lolly A))$ to $A \Lolly (K_B A\iAnd S_B A)$. Since $A$ is strictly stronger than $K_B A \iAnd S_B A$ in general, the inference rule given by $\BRK$ is not invertible. An alternative invertible version of the rule $\BRK$ is
$$\begin{prooftree}
\Gamma \vdash A
\qquad
\Delta, K_B A, K_B A \Lolly A \vdash C
\justifies
\Gamma, \Delta \vdash C
\end{prooftree}
$$
However, the simpler rule is adequate for present purposes.

%

\subsection{Conversion Rules}
\label{sec-conversions}

\begin{figure*}
\begin{align*}
	(\lambda x . t) \, s 
		& \reducby{\betaConv} t[s/x] & \\[2mm]
	\Let{\Pair{x}{y}}{\Pair{t}{u}}{s} 
		& \reducby{\letConv} s[t/x,u/y] & \\[2mm]
	\Break{}{t}{\varphi,f}{s} 
		& {} \reducby{\breakConv}  s[(\lambda p . p\,t) /
\varphi, (\lambda \_ . t) / f ]
\tag*{$
\left(\begin{array}{c@{}l}
            & \varphi \not\in \FV(s) \\
 {} \lor {} & f \not\in \FV(s) \\
 {} \lor {} & \FV(t)  = \emptyset \\
\end{array}\right)$}
\end{align*}
\caption{Standard conversions}
\label{fig-reductions-std}
\end{figure*}

\begin{figure*}
\begin{align*}
(\Let{\Pair{x}{y}}{t}{u})\,s & \reducby{\apLetConv}
\Let{\Pair{x}{y}}{t}{u \, s}  \\[2mm]
\Let{\Pair{v}{w}}{(\Let{\Pair{x}{y}}{t}{u})}{s} &  \reducby{\letLetConv}
\Let{\Pair{x}{y}}{t}{(\Let{\Pair{v}{w}}{u}{s})} \tag*{$(x, y \not\in
\FV(s))$}\\[2mm]
(\Break{}{t}{\varphi,f}{u})\,s & \reducby{\apBreakConv} \Break{}{t}{\varphi,f}{(u\,s)}  \\[2mm]
\Let{\Pair{x}{y}}{(\Break{}{t}{\varphi,f}{u})}{s} &
\reducby{\letBreakConv}
\Break{}{t}{\varphi,f}{(\Let{\Pair{x}{y}}{u}{s})} \tag*{$(\varphi, f
\not\in \FV(s))$}
\end{align*}
\caption{Permuting conversions}
\label{fig-reductions-perm}
\end{figure*}

We now equip our term language with type preserving conversion rules.
The conversion rules we propose are shown in Figures \ref{fig-reductions-std} and \ref{fig-reductions-perm}.
The first two standard conversions are the usual conversions for the affine simply-typed $\lambda$-calculus with pairing.
The third standard conversion shows how to reduce terms involving the new constructor $\Brk$.
The permuting conversions show how to move an occurrence of $\Brk$ or $\Lets$ up one level in the term structure.
We write $t \reduc t'$ when $t'$ can be obtained from $t$ by applying one of the conversion of Figures~\ref{fig-reductions-std} and \ref{fig-reductions-perm} to a single sub-term of $t$. We write $\reducstar$ for the reflexive-transitive closure of $\reduc$.


\section{Example Derivations}

\begin{figure*}
\[
\small
\begin{prooftree}
\[
	x : A \vdash x : A
	\quad
	\[
                	\[
                		\[
                			g : B \Lolly A \vdash g : B \Lolly A
                			\quad
                			\varphi : K_B A \vdash \varphi : K_B A
                			\justifies
                			g : B \Lolly A, \varphi : K_B A \vdash \varphi(g) : B
                			\using{\LE}
                		\]
                		f : B \Lolly A \vdash f : B \Lolly A
                		\justifies
                		g : B \Lolly A, \varphi : K_B A, f : B \Lolly A \vdash \varphi(g) \iAnd f : B \iAnd (B \Lolly A)
                		\using{\CI}
                	\]
		\justifies
		\varphi : K_B A, f : B \Lolly A \vdash \lambda g . \varphi(g) \iAnd f : (A \Lolly B) \Lolly B \iAnd (B \Lolly A)
		\using{\LI}
	\]
	\justifies
	x : A \vdash \Break{}{x}{\varphi,f}{\lambda g . \varphi(g) \iAnd f } : (A \Lolly B) \Lolly B \iAnd (B \Lolly A)
	\using{\BRK}
\]
\justifies
\vdash \lambda x . \Break{}{x}{\varphi,f}{\lambda g . \varphi(g) \iAnd f } : A \Lolly (A \Lolly B) \Lolly B \iAnd (B \Lolly A)
\using{\LI}
\end{prooftree}
\]
\caption{Sample derivation in $\systemB$}
\label{fig-examples}
\end{figure*}


Before investigating the theory of the conversions introduced above, we will first look at some examples of type derivations and some examples of term normalisation using the conversions. For the examples we will use some important $\LLm$ provable formulas from our work and the literature.

\subsection{The divisibility axiom}\label{subsec-divisibility-axiom}

Consider an example of a type whose corresponding formula is not provable in affine logic, but which is provable in $\LLm$, namely
\[ A \Lolly (A \Lolly B) \Lolly (B \iAnd (B \Lolly A)) \] 
This is essentially axiom (B4), and is normally referred to as the \emph{divisibility axiom}. It is well-known that the divisibility axiom characterises {\L}ukasiewicz logic over affine logic \cite{arthan-oliva14a}. 

We can build a term (which can be seen as a proof) having the above type
in the calculus $\systemB$ is as follows: Given $x \colon A$ we can
break it into $\varphi : A \Lolly B) \Lolly B$ and $f : B \Lolly A$.
Using these and $g : A \Lolly B$ we can build a term of type $B \iAnd (B
\Lolly A)$ as $\varphi(g) \iAnd f$. The full type derivation for the
inhabitant
\[ t \;\equiv\; \lambda x . \Break{}{x}{\varphi,f}{\lambda g . \varphi(g) \iAnd f } \]
of $A \Lolly (A \Lolly B) \Lolly B \iAnd (B \Lolly A)$
is shown in Figure \ref{fig-examples}. 

Clearly $t$ is in normal form, i.e., no conversion rules apply to it. However,
using $t$ we can construct the term, 
\[ u \;\equiv\; \lambda x' \lambda g' . \Let{m \iAnd n}{t(x')(g')}{m} \]
of type $A \Lolly (A \Lolly B) \Lolly B$, a type that is already
inhabitited in minimal affine logic. Hence, it is reasonable to ask
whether the conversions in $\systemB$ will reduce this term to a term
without the $\Brk$ constructor. This is indeed the case: First we have:
\begin{align*}
u   &\reducby{\betaConv} 
        \lambda x' \lambda g' . \Let{m \iAnd n}
            {(\Break{}{x'}{\varphi,f}{\varphi(g') \iAnd f }) }{m} \\
    & \reducby{\letBreakConv} 
            \lambda x' \lambda g' . \Break{}{x'}{\varphi, f}{(\Let{m \iAnd n}{\varphi(g') \iAnd f}{m})}\\
   & \reducby{\letConv}
            \lambda x' \lambda g' . \Break{}{x'}{\varphi, f}{\varphi(g')} \equiv v \mbox{, say.}
\end{align*}
%
%
%
Now, in the term $v$, the variable $f$ no longer appears free in the body of the $\Brk$ term, so the side-conditions of
$(\breakConv)$ hold and we may continue as follow to get the normal form for $u$, which does not involve $\Brk$.
\begin{align*}
v   & \reducby{\breakConv} \lambda x' \lambda g' . (\lambda p . p(x'))(g') \\
    & \reducby{\betaConv}  \lambda x' \lambda g' . g'(x').
\end{align*}

In this case we were able to reduce a term with a minimal affine type into a term without $\Brk$ sub-terms, but this is not possible in general. The new $\Brk$ constructor will give rise to new proofs of affine minimal logic theorems. For instance, we have the following normal form proof of identity $A \Lolly A$:
\[ \lambda x^A . \Break{}{x^A}{\varphi^{K_A A}, f^{S_A A}}{\varphi(f)}. \]
 Nevertheless, when this is applied to a closed term $s$ of type $A$, we are able to reduce 
$ (\lambda x . \Break{}{x}{\varphi, f}{\varphi(f)})(s) $
to $s$:
\begin{align*}
\Break{}{s}{\varphi, f}{\varphi(f)}
	& \reducby{\breakConv} (\lambda p . p(s))(\lambda \_ . s) \\[2mm]
	& \reducby{\betaConv} (\lambda \_ . s)(s) \\[2mm]
	& \reducby{\betaConv} s
\end{align*}
so that the new term also behaves like the identity function.

\subsection{Axiom L}

Consider another formula which is provable in basic logic (without using pre-linearity) but is not provable in affine logic, namely the axiom L of \cite{blok-ferreirim00}:
\[ ((B \Lolly A) \Lolly (A \Lolly B)) \Lolly (A \Lolly B) \]
Assuming $\Delta^{(B \Lolly A) \Lolly (A \Lolly B)}$ and $x^A$, we can break $x$ as $\varphi^{K_B A}$ and $f^{B \Lolly A}$ and construct a term of type $B$ as $\varphi(\Delta(f))$. Discharging the two assumptions, in our system we obtain:
\[ \vdash \lambda \Delta \lambda x . \Break{}{x}{\varphi,f}{\varphi(\Delta(f))} : ((B \Lolly A) \Lolly (A \Lolly B)) \Lolly (A \Lolly B) \]
If we take $A \equiv B$, and $\Delta(g) = g$, the term above reduces to
\[ \vdash \lambda x . \Break{}{x}{\varphi,f}{\varphi(f)} : A \Lolly A \]
which, as we have seen in the previous sub-section, is in normal form and behaves as the identity function on each closed term $t^A$.

\subsection{A homomorphism property}

Ferreirim \cite{Ferreirim92} proved an algebraic result (in the algebra of
hoops) suggesting that the following formula should be provable in $\LLm$:
\[ (A \Lolly A \iAnd A) \Lolly (A \Lolly B \iAnd C) \Lolly ((A \Lolly B) \iAnd (A \Lolly C)) \]
Her proof used model-theoretic methods and proved validity of the
formula for a restricted class of algebras.
This constitutes the main lemma in the proof that 
the mapping $X \; \mapsto \; A \Lolly X$ is a hoop homomorphism for idempotent elements $A$. With the
assistance of the Otter system \cite{McCune03} and Veroff's method of proof
sketches \cite{Veroff01}, Veroff and Spinks \cite{Veroff-Spinks04} found a
syntactic proof of the theorem in full generality.  An indirect proof
of the general result using algebraic methods is given in~\cite{arthan-oliva14b}.  Here we
present a term of $\systemB$ with the above type.

Assuming $\alpha \colon A \Lolly A \iAnd A$ and $h \colon A \Lolly B \iAnd C$ we build a term of type $(A \Lolly B) \iAnd (A \Lolly C)$. This term will be built using
\[ \varphi \colon K_{A \Lolly B} (A \Lolly B \iAnd C) \quad \; f \colon S_{A \Lolly B} (A \Lolly B \iAnd C) \]
which we will obtain by breaking $h \colon A \Lolly B \iAnd C$. 

We will define a series of terms $t_1[\varphi], t_2[x,f], \ldots, t_9[h,\alpha]$ where we have listed the free-variables of each term in the brackets. The final term $t_9[h,\alpha]$ will satisfy $$t_9[h,\alpha] \colon (A \Lolly B) \iAnd (A \Lolly C)$$ so that the term $\lambda \alpha \lambda h . t_9[h,\alpha]$ will witness the provability of our homomorphism property. 

Let $\pi_0 \colon B \otimes C \to B$ and $\pi_1 \colon (A \Lolly B) \otimes (A \Lolly C) \Lolly A \Lolly C$ be two closed terms of the indicated types (such terms are easy to define in $\systemB$). Now, we begin the definition of the $t_i$'s:
\begin{itemize}
	\item $t_1[\varphi] \equiv \varphi(\lambda m^{A \Lolly B \iAnd C} \lambda x^A . \pi_0(m \, x))$
	\item $t_2[x^A,f] \equiv \lambda j^{A \Lolly B} . \Let{x' \iAnd y'}{f j x}{(\lambda \_ . x') \iAnd (\lambda \_ . y')}$
\end{itemize}
so that
\begin{itemize}
	\item $t_1[\varphi] \; \colon A \Lolly B$
	\item $t_2[x^A, f] \; \colon (A \Lolly B) \Lolly ((A \Lolly B) \iAnd (A \Lolly C))$
\end{itemize}
Let $Y \equiv (A \Lolly B) \Lolly ((A \Lolly B) \iAnd (A \Lolly C))$. Using $t_2[x,f]$ we build
\[ t_3[x^A, f, p^{Y \Lolly (A \Lolly C)}] \equiv p(t_2[x,f]) \]
so $t_3[x^A, f, p^{Y \Lolly (A \Lolly C)}] \colon A \Lolly C$. Then, using $\alpha \colon A \Lolly A \iAnd A$, we can get a term of type $K_{A \Lolly C} Y$ as
\[ t_4[\alpha, f] \equiv \lambda p^{Y \Lolly (A \Lolly C)} \lambda y^A . \Let{y_0 \iAnd y_1}{\alpha y}{t_3[y_0, f, p](y_1)} \]
So, in summary, we have built two terms
\begin{itemize}
	\item $t_1[\varphi] \; \colon A \Lolly B$
	\item $t_4[\alpha, f] \colon (Y \Lolly (A \Lolly C)) \Lolly A \Lolly C$
\end{itemize}
Now we use $t_1[\varphi]$ to build the term
\[ t_5[\varphi] \equiv \lambda q^Y . q(t_1[\varphi]) \; \; \colon \; \; \underbrace{Y \Lolly (A \Lolly B) \iAnd (A \Lolly C)}_{\equiv Z} \]
We then break $t_5[\varphi] \colon Z$ into 
\[ \psi \colon K_{Y \Lolly (A \Lolly C)} Z \quad \quad g \colon S_{Y \Lolly (A \Lolly C)} Z \]
Defining $t_6 \colon Z \Lolly (Y \Lolly (A \Lolly C))$ as the closed term
\[ t_6 \equiv \lambda u^Z \lambda v^Y . \pi_1(u \, v) \]
we have that $\psi(t_6) \colon Y \Lolly (A \Lolly C)$. Finally, using $g \colon (Y \Lolly (A \Lolly C)) \Lolly Z$ and
\[ t_7 \equiv \lambda v^{A \Lolly C} \lambda u^{A \Lolly B} . u \iAnd v \quad \colon \; (A \Lolly C) \Lolly Y \]
we build
\[ t_8[g] \equiv \lambda i^{A \Lolly C} . \Break{Y}{i}{\eta, k}{g(k)(\eta(t_7))} \]
of type $(A \Lolly C) \Lolly ((A \Lolly B) \iAnd (A \Lolly C))$. The final term $t_9[h,\alpha] \colon (A \Lolly B) \iAnd (A \Lolly C)$ can then be built as
\[ t_9[h, \alpha] \equiv \Break{}{h}{\varphi,f}{(\Break{}{t_5[\varphi]}{\psi,g}{t_8[g](t_4[\alpha,f](\psi(t_6)))})} \]
%

\section{Properties of the Calculus}

In this section we prove three important properties of the calculus $\systemB$: the subject reduction property (conversions preserve types), strong normalisation and the Church-Rosser property.

\subsection{Subject reduction}

We now demonstrate that the $\systemB$ conversions proposed in Figures \ref{fig-reductions-std} and \ref{fig-reductions-perm} all preserve types.

\begin{Theorem} If $t \colon A$ and $t \reduc t'$ then $t' \colon A$.
\end{Theorem}
\Proof It is sufficient to consider each of the conversions applied to the top level of the term $t$. This is clearly the case for the standard $\betaConv$ and $\letConv$. Assume we have a type derivation for $\Break{}{t}{\varphi,f}{s}$. Consider first the case when $t$ is closed, namely
\[
\begin{prooftree}
\[
	\using
	\pi_1
	\leadsto
	\vdash t \colon A
\]
\[
	\varphi \colon K_B A \vdash \varphi \colon K_B A
	\quad
	f \colon S_B A \vdash f \colon S_B A	
	\using
	\pi_2
	\leadsto
	\Delta, \varphi \colon K_B A, f \colon S_B A \vdash s \colon C
\]
\justifies
\Delta \vdash \Break{}{t}{\varphi,f}{s} \colon C
\using{\BRK}
\end{prooftree}
\]
with the two sub-derivations $\pi_1$ and $\pi_2$. The conversion $\breakConv$ in this case corresponds to the following proof transformation
\[
\begin{prooftree}
\[
        \[
            	\using
            	\pi_1
            	\leadsto
            	\vdash t \colon A
        \]
        \justifies
        \vdash \lambda p . p(t) \colon K_B A
\]
\[
        \[
            	\using
            	\pi_1
            	\leadsto
            	\vdash t \colon A
        \]
        \justifies
        \vdash \lambda \_ . t \colon S_B A
\]
\using
\pi_2
\leadsto
\Delta \vdash s[\lambda p . p(t)/\varphi, \lambda \_ . t / f] \colon C
\end{prooftree}
\]
If $t$ is not closed, but either $\varphi$ or $f$ is not free in $s$, the argument is similar but an extra context $\Gamma$ might be present in the derivation $\Gamma \vdash t \colon A$ but since this derivation only needs to be used once this does not invalidate the proof transformation.

Each of the permuting conversions needs to be checked as well, but this is an easy exercise, e.g. for $\apBreakConv$ we are transforming the derivation
\[
\begin{prooftree}
	\[
            \[
            	\using
            	\pi_1
            	\leadsto
            	\Gamma \vdash t \colon A
            \]
            \[
            	\using
            	\pi_2
            	\leadsto
            	\Delta, \varphi \colon K_B A, f \colon S_B A \vdash u \colon C \Lolly D
            \]
            \justifies
            \Gamma, \Delta \vdash \Break{}{t}{\varphi,f}{u} \colon C \Lolly D
         \]
         \[
         	\using 
		\pi_3
		\leadsto
	         \Theta \vdash s \colon C
         \]
         \justifies 
         \Gamma, \Delta, \Theta \vdash (\Break{}{t}{\varphi,f}{u})(s) \colon D
\end{prooftree}
\]
into 
\[
\begin{prooftree}
        \[
            	\using
            	\pi_1
            	\leadsto
            	\Gamma \vdash t \colon A
        \]
	\[
             \[
             	\using 
    		\pi_3
    		\leadsto
    	         \Theta \vdash s \colon C
             \]
            \[
            	\using
            	\pi_2
            	\leadsto
            	\Delta, \varphi \colon K_B A, f \colon S_B A \vdash u \colon C \Lolly D
            \]
            \justifies
            \Theta, \Delta, \varphi \colon K_B A, f \colon S_B A \vdash u \, s \colon D
         \]
         \justifies 
         \Gamma, \Theta, \Delta \vdash \Break{}{t}{\varphi,f}{u \, s} \colon D
\end{prooftree}
\]
by moving $\LE$ above the application of $\BRK$
\Done

\subsection{Strong normalisation}

Let us now prove that the system $\systemB$ is strongly normalising. 

\begin{Definition} Let us call a $\letConv$ or $\breakConv$ conversion in which the variables being substituted do not actually appear free in the term $s$ a \emph{silent} conversion. 
\end{Definition}

We first prove that the set of permuting conversions together with the silent $\letConv$ and $\breakConv$ conversions is strongly normalising:

\begin{Lemma} \label{SN-lemma1} There is no infinite sequence of terms $(t_i)_{i \in \NN}$ such that each $t_{i+1}$ is obtained from $t_i$ by means of a permuting conversion or a silent $\letConv$ or $\breakConv$ conversion.
\end{Lemma}
\Proof The silent conversions make the resulting term strictly smaller than the original one. In a $\Lets$ or a $\Brk$ term
\[ \Let{\Pair{x}{y}}{t}{s} \quad \quad \Break{}{t}{\varphi,f}{s} \]
let us call $t$ the \emph{first} argument, and $s$ the \emph{second} argument. The permuting conversions $\apLetConv$ and $\apBreakConv$ reduce the type complexity of the second argument, e.g. in 
\[ (\Let{\Pair{x}{y}}{t}{s})\,u \]
the term $s$ will have some type $A \Lolly B$, but after the $\apLetConv$ conversion we have
\[ \Let{\Pair{x}{y}}{t}{s \, u} \]
where the second argument $s \, u$ has type $B$. Finally, the permuting conversions $\letLetConv$ and $\letBreakConv$ reduce the size of the first argument for the $\Lets$ or $\Brk$ expressions. If we take the product of these three measures with a lexicographical ordering we obtain a measure which decreases (on a well-founded ordering) after each of these conversions. \Done

Our proof of strong normalisation will make use of the following translation of $\systemB$ terms into terms in the simply typed $\lambda$-calculus with pairing, which we will denote by $\LambdaP$.

\begin{Definition} Define a translation of $\systemB$ terms into $\LambdaP$ terms inductively as:
\[
\begin{array}{lcl}
(x)^* & = & x \\[2mm]
(\lambda x . t)^* & = & \lambda x . t^* \\[2mm]
(s \, t)^* & = & s^* \, t^* \\[2mm]
(\Let{\Pair{x}{y}}{s}{u})^* & = & u^*[\pi_0(s^*)/x][\pi_1(s^*)/y] \\[2mm]
(s \iAnd t)^* & = & s^* \iAnd t^* \\[2mm]
(\Break{}{s}{\varphi,f}{u})^* & = & u^*[k_0(s^*)/\varphi][k_1(s^*)/f]
\end{array}
\]
where $\pi_0, \pi_1$ are the $\LambdaP$ projections, and $k_0 = \lambda x \lambda p . p x$ and $k_1 = \lambda x \lambda \_ \, . x$. 
\end{Definition}

First, it is easy to prove by structural induction on the term $s$ that the translation $s \mapsto s^*$ commutes with substitution:

\begin{Lemma} $(s[t/x])^* = s^*[t^*/x]$
\end{Lemma}

Using the lemma above we can state precisely how the translation of $\systemB$ terms maps to a translation of conversions:

\begin{Lemma} \label{SN-lemma2} We have that:
\begin{itemize}
	\item[(i)] If $t \reduc t'$ via a \emph{non-silent standard conversion} in $\systemB$ then $t^* \reduc^* (t')^*$ in one or more standard conversions in $\LambdaP$. 
	\item[(ii)] If $t \reduc t'$ via a \emph{silent standard conversion} or a \emph{permuting conversion} in $\systemB$ then $t^* = (t')^*$.
\end{itemize}
\end{Lemma}

\begin{Theorem}\label{thm-sn-standard} $\systemB$ is strongly normalising.
\end{Theorem}

\Proof Suppose that there was an infinite sequence $(t_i)_{i \in \NN}$ of $\systemB$ terms such that for each $i$ we have that $t_{i+1}$ is obtained from $t_i$ by one of the $\systemB$ conversions (standard or permuting). By Lemma \ref{SN-lemma2} we would then obtain a sequence of $\LambdaP$-terms $(t_i^*)_{i \in \NN}$ where for each $i$, either
\begin{itemize}
	\item $t^*_{i+1}$ is obtained from $t_i^*$ via one or more $\LambdaP$ conversions, or
	\item $t^*_{i+1} = t^*_i$
\end{itemize}
Since $\LambdaP$ is strongly normalising, we know that from some number $N$ and all $i \geq N$ we must have that $t^*_i = t^*_{i+1}$. But this would mean that in the original sequence, we have an infinite chain of permuting conversions or silent standard conversions, contradicting Lemma \ref{SN-lemma1}.
\Done

\subsection{Church-Rosser property}

\begin{Theorem}\label{thm-church-rosser}
$\systemB$ has the Church-Rosser property.
\end{Theorem}
\Proof
Since we have strong normalisation for $\systemB$, by Newman's lemma, it is enough to prove the weak Church-Rosser property, i.e.,
that if $w \reduc w'$ and $w \reduc w''$ then there is
a $w'''$ such that $w'' \reducstar w'''$ and $w'' \reducstar w'''$.
The proof is by induction on the structure of $w$ and is fairly standard,
so we will only sketch it here.
See \cite[Theorem 6.3.9]{sorensen-urzyczyn06} for an example of this
kind of proof.
One checks that for all substitutions $\sigma$ and all terms $s$ and $s'$
\begin{itemize}
	\item if $s \reduc s'$ then $s[\sigma]\reduc  s'[\sigma]$, and
	\item for all total functions $f \subseteq \reducstar$ we have $s[\sigma] \reducstar s[\sigma']$, where $\sigma' = f \circ \sigma$.
\end{itemize}
These facts deal with the only tricky case in the proof
for the simply-typed $\lambda$-calculus, when $w$ is the
$\beta$-redex $(\lambda x . s)u$, $w' = s[u/x]$ and $w''$ is either $(\lambda x . s')u$
or $(\lambda x . s)u'$. The proof now reduces
to an analysis of the {\em critical pairs}:
i.e., the reducts $w'$ and $w''$  of a term $w$
in which two redexes overlap in such a way that carrying
out either conversion affects the structure required by the other conversion.
Inspection of the conversions shows that the are no critical pairs
involving $\betaConv$, but critical pairs do arise for the following
pairs of conversions.
$$
\begin{array}{|r@{\mbox{ v. }}l|}\hline
\letConv
   & \apLetConv \\\hline
\letConv
   & \letLetConv \\\hline
\breakConv
   &  \apBreakConv \\\hline
\breakConv
   & \letBreakConv \\\hline
\apLetConv
   & \letLetConv \\\hline
\apLetConv
   & \letBreakConv \\\hline
\letLetConv
   & \letLetConv \\\hline
\end{array}
$$
In the first four types of critical pair, the conversion to $w'$ (say)
eliminates both the redexes while $w''$ still has a redex of the
same type as was used to reduce $w$ to $w'$.
E.g., consider
the critical pair of the form $\letConv$ v. $\apLetConv$. If we set:
\begin{align*}
w &= (\Let{\Pair{x}{y}}{\Pair{t}{u}}{s})r \\
w' &= (s[t/x,u/y]) r \\
w'' &= \Let{\Pair{x}{y}}{\Pair{t}{u}}{(s\,r)}
\end{align*}
then $w \reducby{\letConv} w'$ and $w \reducby{\apLetConv} w''$.
So taking $w''' = w'$, we have that $w' \reducstar w'''$ (trivially)
and that $w'' \reducby{\letConv} w''$.

In the remaining three types of critical pair, both $w'$ and $w''$
require further conversion.
E.g., consider $\apLetConv$ v. $\letBreakConv$.
If we set:
\begin{align*}
w &= (\Let{\Pair{x}{y}}{(\Break{}{t}{\phi, f}{u})}{s})r \\
w' &= \Let{\Pair{x}{y}}{(\Break{}{t}{\phi, f}{u})}{(s\,r)} \\
w'' &= (\Break{}{t}{\phi, f}{(\Let{\Pair{x}{y}}{u}{s})) r}
\end{align*}
then $w \reducby{\apLetConv} w'$ and $w \reducby{\letBreakConv} w''$.
But then putting
$$w''' = \Break{}{t}{\phi, f}{(\Let{\Pair{x}{y}}{u}{(s\,r)})}$$
we have:
\begin{align*}
w' &\reducby{\letBreakConv}
      \Break{}{t}{\phi, f}{(\Let{\Pair{x}{y}}{u}{(s, r)})} = w'''\\
 &\quad\mbox{and} \\
w'' &\reducby{\apBreakConv}
      \Break{}{t}{\phi, f}{((\Let{\Pair{x}{y}}{u}{s})r)} \\
   &\reducby{\apLetConv}
      \Break{}{t}{\phi, f}{(\Let{\Pair{x}{y}}{u}{(s, r)})} = w'''\\
\end{align*}
The treatment of the other types of critical pair is similar.
\Done

Although the Church-Rosser property is not difficult to prove, it
was quite tricky to find a suitable system of conversions. One of
our earlier attempts included the following conversion
$$
\Break{}{(\Let{\Pair{x}{y}}{t}{u})}{\varphi,f}{s} 
\;\reducby{\breakLetConv}\; 
\Let{\Pair{x}{y}}{t}{(\Break{}{u}{\varphi, f}{s})}
$$
If we put:
\[
w = \Break{}{\Let{\Pair{x}{y}}{t}{u}}{\phi, f}{s}
\]
then we find (assuming $\FV(t) = \FV(u) = \emptyset$) that:
\begin{align*}
w & \reducby{\breakConv}
   s [\lambda p . p(\Let{\Pair{x}{y}}{t}{u})/\phi, \lambda \_ . \Let{\Pair{x}{y}}{t}{u}/f]\\
 &\quad\mbox{and} \\
w & \reducby{\breakLetConv}
      \Let{\Pair{x}{y}}{t}{(\Break{}{u}{\phi, f}{s})} \\
   &  \reducby{\breakConv}
      \Let{\Pair{x}{y}}{t}s[\lambda p . p\,u/\phi,\lambda \_ . u/f]
\end{align*}
So $w$ would have two distinct normal forms if we admitted $\breakLetConv$.

\section{A Gentzen-style Calculus and Cut Elimination}

We have so far discussed the dynamics of $\LLm$ natural deduction proofs via \emph{normalisation}, using our Curry-Howard correspondence. We can also look at these results using a Gentzen-style calculus with left and right rules, and look into \emph{cut elimination}. Such a system for $\LLm$ is given in Figure \ref{fig-sequent-calculus}.

\begin{Theorem} The provable sequents of the Gentzen system of Figure \ref{fig-sequent-calculus} are the same as those of $\LLm$.
\end{Theorem}

\begin{proof} The break rule of Figure \ref{fig-sequent-calculus} matches precisely the break rule of our natural deduction system $\systemB$, which we have already shown to coincide with $\LLm$. It is standard to show that the left and right rules are inter-derivable with the introduction and elmination rules of the natural deduction system.
\end{proof}

\begin{figure*}
\[
\begin{array}{ccc}
%
%
\multicolumn{3}{c}{
\begin{prooftree}
\justifies
\Gamma, A \vdash A
\using{\ASM}
\end{prooftree}
}
\\[7mm]
%
%
\begin{prooftree}
\Gamma \vdash  A
\qquad
\Delta, A \vdash C
\justifies
\Gamma, \Delta \vdash C
\using{\CUT}
\end{prooftree}
& \quad &
\begin{prooftree}
\Gamma \vdash  A
\qquad
\Delta, K_B A, S_B A \vdash C
\justifies
\Gamma, \Delta \vdash C
\using{\BRK}
\end{prooftree}
\\[7mm]
%
%
\begin{prooftree}
\Gamma, A \vdash B
\justifies
\Gamma \vdash A \Lolly B
\using{\LRight}
\end{prooftree}
& &
%
%
\begin{prooftree}
\Gamma \vdash A
\quad
\quad
\Delta, B \vdash  C
\quad
\justifies
\Gamma, \Delta, A \Lolly B \vdash C
\using{\LLeft}
\end{prooftree} \\[7mm]
%
%
\begin{prooftree}
\Gamma \vdash A
\quad
\Delta \vdash B
\justifies
\Gamma, \Delta \vdash A \iAnd B
\using{\CRight}
\end{prooftree}
& &
%
%
\begin{prooftree}
\Gamma, A, B \vdash u : C
\justifies
\Gamma, A \iAnd B \vdash C
\using{\CLeft}
\end{prooftree} 
\end{array}
\]
\caption{Sequent calculus for $\LLm$}
\label{fig-sequent-calculus}
\end{figure*}

\begin{Theorem} The cut rule \CUT{} is eliminable from the proof system of Figure \ref{fig-sequent-calculus}.
\end{Theorem}

\begin{proof} For each left and right rule, let us call the formula which is being introduced the \emph{major} formula. When the application of cut involves two major formulas, then such cut can be replaced by cuts of smaller complexity, as in the standard cut elimination procedure. In all other cases, when the cut formula $A$ is not major, the cut can be pushed up the proof tree. This can also be done with the new break rule \BRK{}. When one of the premises of the cut rule is an axiom the cut rule can be eliminated.
\end{proof}

The reader might have noticed, however, that \BRK{} has a very similar flavour to \CUT{}. But it follows directly from the theorem above that the rule \BRK{} is not derivable from \CUT{}, since \CUT{} is eliminable but \BRK{} is not. 
%
%
There are, however, some particular instance of \BRK{} which are indeed derivable from \CUT{}.

\begin{Theorem} When $\Gamma = \emptyset$ then \BRK{} is derivable from \CUT{}.
\end{Theorem}

\begin{proof} We can derive \BRK{} as follows:
\[
\begin{prooftree}
\[
	\vdash  A
	\Justifies
	\vdash B \Lolly A
\]
\[
    \[
    	  \vdash  A
    	  \Justifies
    	  \vdash K_B A
    \]
    \qquad
    \Delta, K_B A, S_B A \vdash C
    \justifies
    \Delta, B \Lolly A \vdash C
    \using{\CUT}
\]
\justifies
\Delta \vdash C
\using{\CUT}
\end{prooftree}
\]
where the double lines indicate one or more steps.
\end{proof}

\begin{Theorem} When $K_B A$ or $S_B A$ is a superfluous assumption in proving $\Delta, K_B A, S_B A \vdash C$ then \BRK{} is derivable from \CUT{}.
\end{Theorem}

\begin{proof} Consider the case when $K_B A$ is not needed, so that we actually have $\Delta, S_B A \vdash C$. Such instance of \BRK{} can be derived from \CUT{} as: 
\[
\begin{prooftree}
\[
	  \vdash  A
	  \Justifies
	  \vdash B \Lolly A
\]
\qquad
\Delta, B \Lolly A \vdash C
\justifies
\Delta \vdash C
\using{\CUT}
\end{prooftree}
\]
The case where $S_B A$ is superfluous can be treated in the same way.
\end{proof}

These last two theorems justify our side conditions for the conversion rule (\breakConv) from Section \ref{sec-conversions}. The context $\Gamma$ being empty corresponds to the term $t$ being closed ($\FV(t) = \emptyset$), whereas $K_B A$ or $B \Lolly A$ begin superfluous assumptions correspond to $\varphi \not\in \FV(s)$ or $f \not\in \FV(s)$. From a Gentzen-style point of view, these are the cases where we can always replace a break rule by a standard cut rule.

\section{Concluding Remarks}

Strong normalisation for the standard simply-typed $\lambda$-calculus with pairing is well-known.
Strong normalisation for the affine fragment of that calculus follows or can be proved more directly when one observes that the conversions always reduce the size of an affine term.
Troelstra \cite{Troelstra95} proved strong normalisation for a variant of the term calculus for intuitionistic linear logic proposed by Benton et al. \cite{Benton(1993)}.
In that calculus it is the exponential operator $!$ that makes the normalisation result tricky, since the usual
introduction rule for $!$ also acts as an elimination rule.
Similarly, in $\systemB$, the rule for $\Brk$ complicates the normalisation proof.
In both cases, the desire to control contraction is the source of the difficulty.

The decision problem for classical {\L}ukasiewicz logic is known to be co-NP-complete while the
decision problem for minimal {\L}ukasiewicz logic can be shown to reduce to the decision
problem for the equational theory of commutative GBL-algebras, which is known to be PSPACE-complete \cite{Bova:2009}.
In both cases, the known decision procedures are based on semantic methods and no effective
 proof search methods are known. The present work is motivated by a desire either to find
such algorithms or to understand why they cannot exist.
It seems highly unlikely that
a logic with a PSPACE-complete decision problem could admit an analytic inference system.
However from the strong normalisation property, one can hope to derive effective bounds on the size of a deduction and the formulas in it and so, perhaps, find some weak form of the sub-formula property that could enable a proof-theoretic decision procedure.


\bibliographystyle{plain}

\bibliography{../references}

%

\end{document}